\documentclass[a4paper,reqno, 11pt]{amsart}
\usepackage{fullpage}

\usepackage{graphicx,color,marginnote}
\usepackage{amsmath}
\usepackage[utf8]{inputenc}
\usepackage[T1]{fontenc}
\usepackage{cite}
\usepackage[english]{babel}
\usepackage[normalem]{ulem} 

\usepackage{amsfonts}
\usepackage{amssymb}
\usepackage{bbm}
\usepackage{epstopdf}
\usepackage{color}
\numberwithin{equation}{section}

\usepackage{amsthm}
\theoremstyle{plain}
\newtheorem{theorem}{Theorem}[section]

\newtheorem{corollary}[theorem]{Corollary}

\newtheorem*{assumption*}{Assumption}

\newtheorem{lemma}[theorem]{Lemma}
\newtheorem{definition}[theorem]{Definition}
\newtheorem{proposition}[theorem]{Proposition}

\theoremstyle{remark}
\newtheorem{remark}[theorem]{Remark}

\newcommand{\old}[1]{}

\newcommand{\G}{\mathcal{G}}
\newcommand{\T}{\mathcal{T}}
\newcommand{\Or}{\mathcal{O}}

\newcommand{\ZZ}{\mathbb{Z}}

\newcommand{\CC}{\mathbb{C}}

\newcommand{\EE}{\mathbb{E}}

\renewcommand{\i}{\mathrm{i}}

\DeclareMathOperator{\Log}{Log}

\usepackage{bm}
\usepackage{color}

\title[Aztec diamond and t-embeddings]{Perfect t-embeddings and the octahedron equation of the two-periodic Aztec diamond}

\author[Tomas Berggren]{Tomas Berggren$^\mathrm{a}$}

\author[Marianna Russkikh]{Marianna Russkikh$^\mathrm{b}$}

\thanks{\textsc{${}^\mathrm{A}$ Royal Institute of Technology, Department of Mathematics, 100 44 Stockholm, Sweden}}
\thanks{\textsc{${}^\mathrm{B}$  University of Notre Dame, Department of Mathematics, 255 Hurley Bldg,
Notre Dame, IN 46556, United States of America}}

\thanks{\texttt{tobergg@kth.se}, \texttt{mrusskik@nd.edu}}

\begin{document}

\begin{abstract} This paper explores the connection between perfect t-embeddings and the octahedron equation in the setting of the two-periodic Aztec diamond. In particular, we show that the positions of both the t-embedding and the corresponding origami map can be expressed as sums of density functions arising from solutions to the octahedron equation with appropriate flat initial conditions.
\end{abstract}

\keywords{dimer model, aztec diamond, perfect t-embeddings, octahedron equation}


\maketitle


\section{Introduction}
Our interest in studying \emph{perfect t-embeddings} and their associated \emph{origami maps} stems from their potential to establish the scaling limit of height fluctuations in dimer models~\cite{CLR2}. Constructing these embeddings and analyzing their behavior is therefore a compelling and important problem -- one that has been addressed in an increasing number of settings; see~\cite{Ch-R, BNR23, BNR24, BNR25, KV25}.

The \emph{octahedron equation} is a~$2+1$-dimensional discrete integrable system governing the evolution of a function~$T_{j,k,n}$ on the integer lattice. Originating in the study of quantum integrable systems~\cite{KNS94a, KNS94b}, it has since emerged in a variety of combinatorial settings, see e.g.,~\cite{DF13}, and has been linked to the theory of cluster algebras~\cite{DFK09}. Notably, its solutions can be represented as partition functions of dimer models on Aztec diamonds~\cite{Sp, DF}.

It was shown in~\cite{Ch-R} that the perfect t-embeddings of the uniformly weighted Aztec diamond satisfy the \emph{discrete wave equation}. A similar result holds for the \emph{density functions} associated with solutions of the octahedron equation with flat initial conditions~\cite{DF-SG}. 
Both results can be viewed as consequences of the \emph{shuffling algorithm} applied to the Aztec diamond.
This highlights a close structural resemblance between the two constructions. The goal of this paper is to investigate the connection between perfect t-embeddings of the \emph{two-periodic Aztec diamond} and the solutions of the octahedron equation with appropriate initial conditions. 
We show that both satisfy the same discrete wave equation, albeit with slightly different boundary conditions, and that one can be expressed in terms of the other.

More precisely, we demonstrate that, in the setting of the two-periodic Aztec diamond, the perfect t-embedding admits a natural probabilistic interpretation: it can be written as a sum of density functions associated with solutions of the octahedron equation. This result parallels the expression obtained in~\cite{BNR23}, where perfect t-embeddings of the uniformly weighted Aztec diamond were represented in terms of \emph{edge probabilities}.

Our main theorem relates perfect t-embeddings and their origami map of the two-periodic Aztec diamond to the solution of the octahedron equation. We recall here the definition of the two-periodic Aztec diamond and the octahedron equation, and refer the reader to Section~\ref{sec:t-emb} for a definition of perfect t-embeddings and their origami map.
The two-periodic Aztec diamond is the Aztec diamond with a $1$-parameter family of $(2\times 2)$-periodic edge weights and was introduced and studied in \cite{CY14, CJ16}.
Let us define the Aztec diamond~$A_{n}$  of size~$n$ to be a subset of faces~$(j,k)$ of the square grid~$(\mathbb{Z}+\frac12)^2$ such that $|j|+|k| \leq n-1$. Let $e$ be an edge of the Aztec diamond of size~$n$ adjacent to a face $(j,k)$ with $j+k$ odd, we define a $(2\times 2)$-periodic weight function $\nu$ on edges as follows:
\begin{equation*}
\nu_e=
\begin{cases}
a, &n=1,2 \mod 4 \text{ and } j \text{ even} \quad \text{or} \quad n=3,0 \mod 4 \text{ and } j \text{ odd}, \\
1, & \text{ otherwise},
\end{cases}
\end{equation*}
where $a>0.$

For $j, k, n \in\mathbb{Z}$, $n\geq 1$, $j+k+n$ odd, let $T^{\operatorname{oct}}_{j,k,n}\in\mathbb{R}$ be the solution of the \emph{octahedron equation} 
\begin{equation*}
T^{\operatorname{oct}}_{j,k,n+1}T^{\operatorname{oct}}_{j,k,n-1}=
T^{\operatorname{oct}}_{j+1,k,n}T^{\operatorname{oct}}_{j-1,k,n}+
T^{\operatorname{oct}}_{j,k+1,n}T^{\operatorname{oct}}_{j,k-1,n},
\end{equation*}
with  the initial condition
 \begin{equation*}
T^{\operatorname{oct}}_{j, k, n_{jk}} = t_{j,k},
\end{equation*}
where~$t_{j,k}\in\mathbb{R}$ and~$n_{jk}=j+k+1 \mod 2$. For~$(\epsilon, \eta)\in \mathbb{Z}^2, n\geq 0$ let us define the density function~$\rho^{(\epsilon, \eta)}_{j,k,n}$ by
\begin{equation*}
\rho^{(\epsilon, \eta)}_{j,k,n}=t_{\epsilon, \eta}
\partial_{t_{\epsilon, \eta}}
\log{T^{\operatorname{oct}}_{j,k,n}}.
\end{equation*}
In the definition of the density function, we view the initial conditions $t_{\epsilon, \eta}$ as variables even though we will later specify their values. We are now ready to state our main result.

\begin{figure}
 \begin{center}
\includegraphics[width=0.45\textwidth]{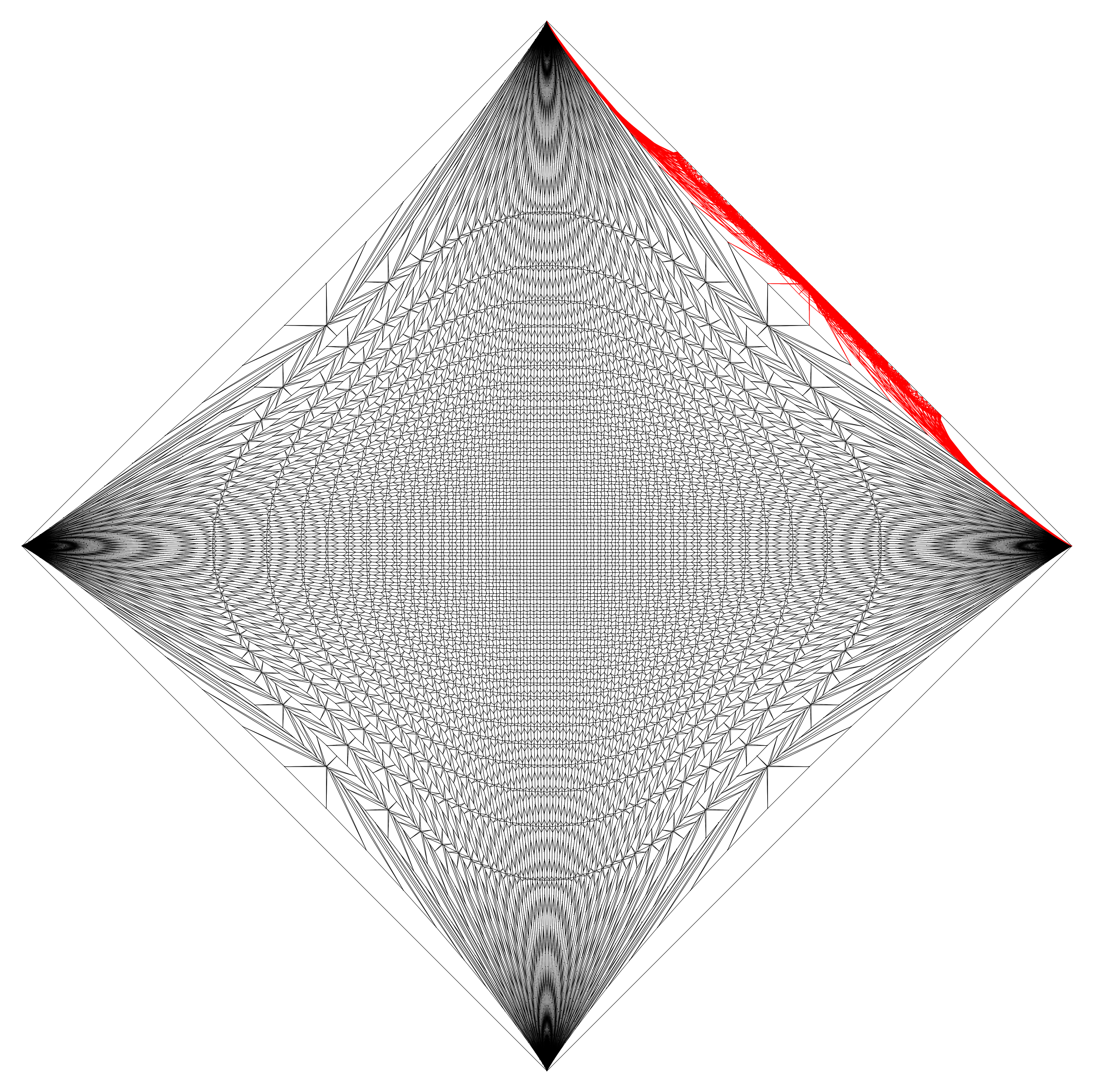}
\includegraphics[width=0.45\textwidth]{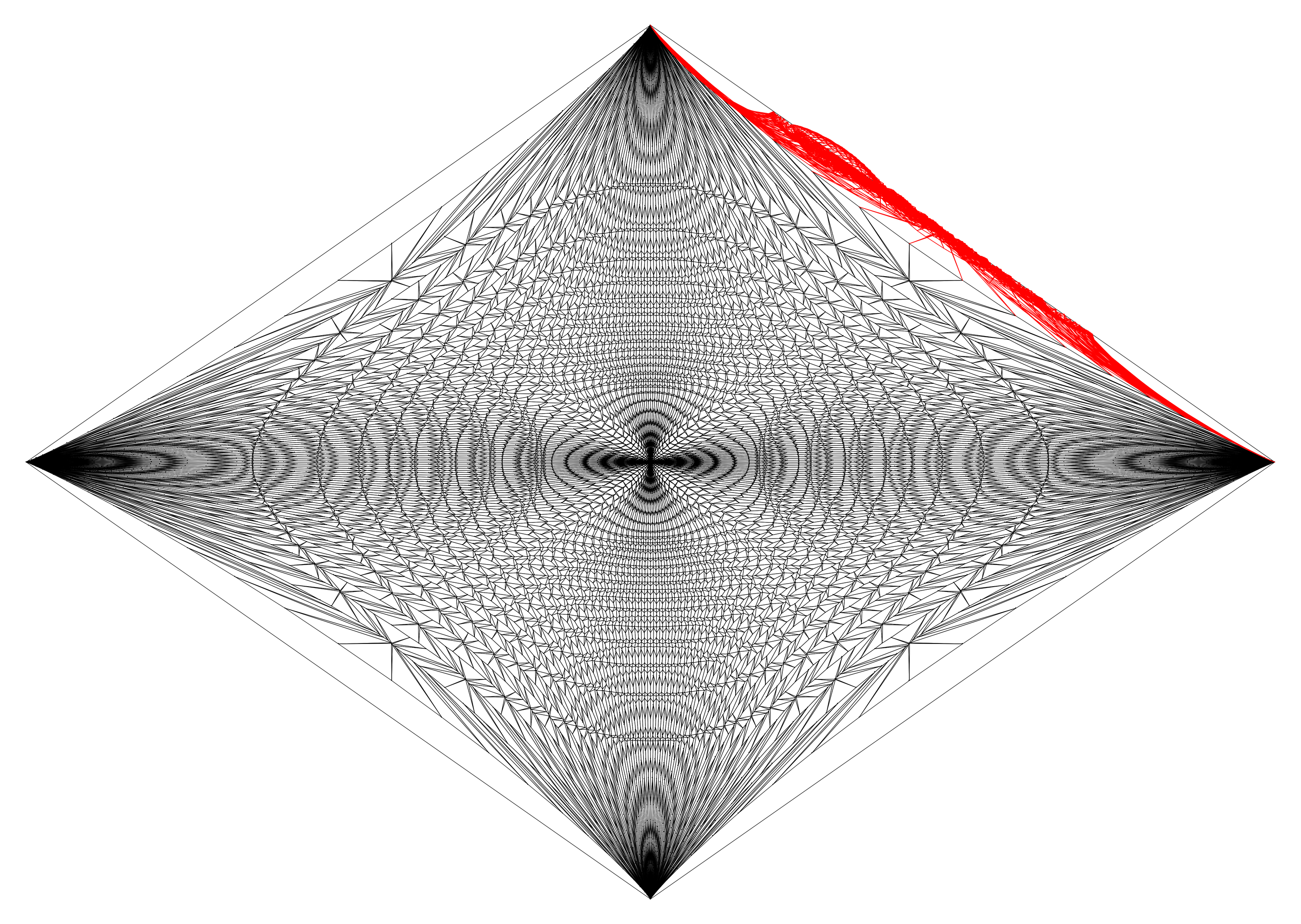}
 \caption{Perfect t-embedding (black) and origami map (red) of the reduced Aztec diamond of size~$100$, constructed via reccurence relations derived from the shuffling algorithm. Left:~$a=1$. Right: $a=0.7$.}\label{fig:t_o}
 \end{center}
\end{figure}

\begin{theorem}\label{thm} 
Let~$A_{n+1}$, for~$n\geq 1$, be a sequences of weighted Aztec diamonds described above with weights determined by the parameter~$a>0$, and let~$\rho^{(\epsilon, \eta)}_{j,k,n}$ be the density function associated with the octahedron equation with initial conditions $t_{j,k}=a^{-1}$ if $j=1, \, k=0 \mod 2$, and $t_{j,k}=1$ otherwise.
Then the functions~$\T_n$ and~$\Or_n$ defined on the faces of~$A_{n+1}$ and given by 
\begin{equation}
\T_n(j,k)=f_E(j,k,n)+\i{ a} f_N(j,k,n)-f_W(j,k,n)-\i { a} f_S(j,k,n),
\end{equation}
\begin{equation}
\Or_n(j,k)=f_E(j,k,n)+\i{ a} f_N(j,k,n)+f_W(j,k,n)+\i { a} f_S(j,k,n),
\end{equation}
define a perfect t-embedding and its origami map of the two-periodic Aztec diamond~$A_{n+1}$. 
Here, the functions~$f_E, f_N, f_W$ and~$f_S$ are given by
\begin{equation*}
\begin{split}
f_E(j,k,n)=
&\tfrac{1}{1+a^{-2}}\sum_{s=0}^{m+1}   \rho^{(0,0)}_{j-4s, \,k,\, n-4s } 
+\tfrac{1}{1+a^{2}}\sum_{s=0}^{m+1} \rho^{(1,1)}_{j-4s-1,\, k+1, \, n-4s-2 } \\
&\quad\quad-\tfrac{1}{2}\sum_{s=0}^{m+1} \rho^{(0,1)}_{j-4s-3, \, k+1,\, n-4s-2 } 
-\tfrac{1}{2}\sum_{s=0}^{m+1}\rho^{(1,0)}_{j-4s, \, k, \, n-4s };
\end{split}
\end{equation*}

\begin{equation*}
\begin{split}
f_W(j,k,n)=
&\tfrac{1}{1+a^{-2}}\sum_{s=0}^{m+1} \rho^{(0,0)}_{j+4s, \, k, \, n-4s}
+\tfrac{1}{1+a^{2}}\sum_{s=0}^{m+1} \rho^{(1,1)}_{j+4s+3, \, k+1, \, n-4s-2}\\
&\quad\quad-\tfrac{1}{2}\sum_{s=0}^{m+1} \rho^{(0,1)}_{j+4s+3, \, k+1, \, n-4s-2}
-\tfrac{1}{2}\sum_{s=0}^{m+1} \rho^{(1,0)}_{j+4s, \, k, \, n-4s};
\end{split}
\end{equation*}

\begin{equation*}
\begin{split}
f_S(j,k,n)=
&\tfrac{1}{1+a^{2}}\sum_{s=0}^{m+1} \rho^{(0,0)}_{ j, \, k+4s, \, n-4s}
+\tfrac{1}{1+a^{-2}}\sum_{s=0}^{m+1} \rho^{(1,1)}_{j+1, \, k+4s+3, \, n-4s-2}\\
&\quad\quad-\tfrac{1}{2}\sum_{s=0}^{m+1} \rho^{(0,1)}_{j+1, \, k+4s+3, \, n-4s-2}
-\tfrac{1}{2}\sum_{s=0}^{m+1} \rho^{(1,0)}_{j-4s, \, k, \, n-4s};
\end{split}
\end{equation*}

\begin{equation*}
\begin{split}
f_N(j,k,n)=
&\tfrac{1}{1+a^{2}}\sum_{s=0}^{m+1} \rho^{(0,0)}_{ j, \, k-4s, \, n-4s}
+\tfrac{1}{1+a^{-2}}\sum_{s=0}^{m+1} \rho^{(1,1)}_{ j+1, \, k-4s-1, \, n-4s-2}\\
&\quad\quad-\tfrac{1}{2}\sum_{s=0}^{m+1} \rho^{(0,1)}_{ j+1, \, k-4s-3, \, n-4s-2}
-\tfrac{1}{2}\sum_{s=0}^{m+1} \rho^{(1,0)}_{ j, \, k-4s, \, n-4s},
\end{split}
\end{equation*}
where~$m=\lfloor\frac{n}{4}\rfloor$.
\end{theorem}
\begin{remark} 
Let~$A_{j,k,n}$ denote the Aztec diamond of size~$n-1$ centered at~$(j,k)$ with edge weight function~$\hat \nu$. For an edge~$e$ adjacent to a face~$(\epsilon, \eta)$ with~$\epsilon+\eta$ odd, set~$\hat \nu_e=a$ if~$\epsilon$ is odd, and~$\hat \nu_e=1$ if~$\epsilon$ is even. It is a remarkable fact that the solution~$T^{\operatorname{oct}}_{j,k,n}$ of the octahedron equation with initial condition given in the previous theorem can be expressed in terms of the partition function of the dimer model on~$A_{j,k,n}$. Consequently, 
\begin{equation*}
\rho^{(\epsilon, \eta)}_{j,k,n}=\EE_{j,k,n}\left[1-\mathcal D_{\epsilon, \eta}\right],
\end{equation*}
where $\mathcal D_{\epsilon, \eta}$ denotes the number of dimers adjacent to the face $(\epsilon, \eta)$ in a random dimer configuration of~$A_{j,k,n}$; see Section~\ref{sec:t-emb_octa}, in particular Remarks~\ref{rem:density_function} and~\ref{rem:weights}, for details. This probabilistic interpretation of the density function casts the right-hand side of the expressions in the previous theorem in a more probabilistic light, in parallel with the results of~\cite{BNR23}.
\end{remark}

The expressions from Theorem~\ref{thm} can be used to study the limit of the perfect t-embeddings and their corresponding origami maps for the two-periodic Aztec diamond as~${n\to\infty}$. We do not provide the technical details of this convergence since,
recently, the method introduced in~\cite{BNR24} was extended to a broad class of doubly periodic Aztec diamonds in~\cite{BNR25}, where the convergence of t-surfaces to maximal surfaces with cusps was established. The formation of the cusp can already be anticipated in Figure~\ref{fig:t_o}. This alternative approach is based on expressing both the perfect t-embedding and the origami map in terms of the inverse Kasteleyn matrix. We expect that method to be more robust than the one presented here, as it relies solely on the inverse Kasteleyn matrix, rather than on the shuffling algorithm. However, there exist models for which explicit formulas for the inverse Kasteleyn matrix are not known, while a connection to the octahedron equation has been established; see~\cite{DF-V}. We believe that the connection between perfect t-embeddings and the octahedron equation demonstrated here for the two-periodic Aztec diamond can be extended to such models, and used to study the scaling limits of their associated t-surfaces.


\addtocontents{toc}{\protect\setcounter{tocdepth}{1}}
\subsection*{Acknowledgements}
We are grateful to Terrence George and Matthew Nicoletti for many stimulating and insightful discussions. 
We thank Leonid Petrov for his help with the simulations.
We would also like to thank Alexei Borodin and Tom Hutchcroft for their support and interest.
TB was supported by the Knut and Alice Wallenberg Foundation grant KAW 2019.0523. 
MR was partially supported by a Simons Foundation Travel Support for Mathematicians MPS-TSM-00007877. 
\addtocontents{toc}{\protect\setcounter{tocdepth}{2}}

\section{Perfect t-embeddings of Aztec diamonds}

\subsection{Perfect t-embeddings and origami maps}\label{sec:t-emb} In this section, we review the basic definitions and properties of \emph{t-embeddings}, also known as \emph{Coulomb gauges}. We refer an interested reader to \cite{KLRR, CLR1, CLR2} for more details.

 Let $(\G, \nu)$ be a weighted planar bipartite graph. 
 We define a probability measure~$\mathbb{P}$ on the set~$\mathcal{M}$ of dimer configurations by
 \[\mathbb{P}[m]=\frac{\prod_{e\in m} \nu(e)} { \sum_{m\in\mathcal{M}} \prod_{e\in m} \nu(e)} .\]
 Recall that two weight functions $\nu, \nu': E(\G)\to \mathbb{R}_{+}$ are \emph{gauge equivalent} if there exists a pair of gauge functions 
$(F^\bullet_{\operatorname{gauge}}, F^\circ_{\operatorname{gauge}})$ 
such that $\nu'(wb)=F^\circ_{\operatorname{gauge}}(w)\nu(wb)F^\bullet_{\operatorname{gauge}}(b)$. 
Given edge weights~$\nu$ on a bipartite graph, one can associate a face weight $X_{v^*}$ to each face of $\G$ by
\[X_{v^*}:=\prod_{s=1}^d\frac{\nu(w_sb_s)}{\nu(w_{s+1}b_s)},\]
where the face $v^*$ has degree $2d$ with vertices denoted by $w_1, b_1, \ldots , w_d, b_d$ in clockwise order. Note that two weight functions are gauge equivalent if and only if they correspond to the same face weights.


\begin{definition}\label{def:temp} Given a weighted planar bipartite graph~$(\G, \nu)$, a t-embedding $\T(\G^*)$ is a proper embedding of an augmented dual graph, where the outer face of $\T(\G^*)$ corresponds to the cycle replacing $f_{\operatorname{out}}$ in the augmented dual $\G^*$ 
such that the following conditions are satisfied:
\begin{itemize}
\item[$\bullet$] the sum of the angles at each inner vertex of $\T(\G^*)$ at the corners corresponding to black faces is equal to $\pi$ (and similarly for white faces),
\item[$\bullet$] the geometric weights (dual edge lengths) $|\T(v^*_1)-\T(v^*_2)|$ are gauge equivalent to $\nu_e$, where $v^*_1$ and $v^*_2$ are vertices of the dual edge $e^*$.
\end{itemize}
\end{definition}

The last condition of the above definition can be equivalently formulated in terms of face weights.

\begin{remark}
In the setup of Definition~\ref{def:temp},
 let $v^*$ be a face of
degree~$2d$ with vertices (in clockwise order)~$w_1, b_1, \ldots , w_d, b_d$. Let~$X_{v^*}$ be the face weight associated with the face~$v^*$ and weight function~$\nu$.
Then the second condition of Definition~\ref{def:temp} is equivalent to 
\[X_{v^*}=\prod_{s=1}^d\frac{|\T(v^*)-\T(v^*_{2s-1})|}{|\T(v^*_{2s})-\T(v^*)|} \quad \text{ for all inner } v^*\in\G^*,\]
where we by~$v^*_1, v^*_2, \ldots, v^*_{2d}$ denote the dual vertices adjacent to~$v^*$, such that~$v^*v^*_{2s-1}=(w_s b_s)^*$
and~$v^*v^*_{2s}=(b_sw_{s+1})^*$.   

Moreover, the angle condition 
 implies that 
\[X_{v^*}=(-1)^{d+1}\prod_{s=1}^d\frac{\T(v^*)-\T(v^*_{2s-1})}{\T(v^*_{2s})-\T(v^*)}.\]
\end{remark}

The notion of \emph{perfect t-embedding} was introduced in~\cite{CLR2}.

\begin{definition}
A t-embedding is perfect if the outer face of $\T(\G^*)$ is a tangential polygon
to a circle and all the non-boundary edges emanating from boundary vertices are bisectors of the corresponding angles
\end{definition}

To each t-embedding $\T(\G^*)$ one can associate the so-called origami map $\Or:\G^*\to\mathbb{C}$.

\begin{definition} To get an origami map $\Or(\G^*)$ from $\T(\G^*)$ one can choose a white root
face $w_0$, set $\Or(v^*)=\T(v^*)$ for all vertices $v^*$ adjacent to the root face, and fold the plane along every edge of the t-embedding.
\end{definition}

\begin{figure}
 \begin{center}
\includegraphics[width=0.17\textwidth]{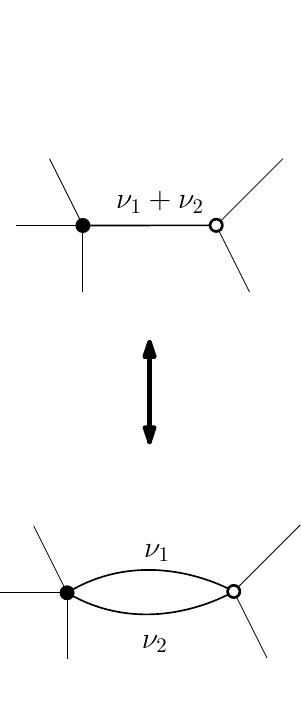}
$\quad\quad\quad\quad$
\includegraphics[width=0.17\textwidth]{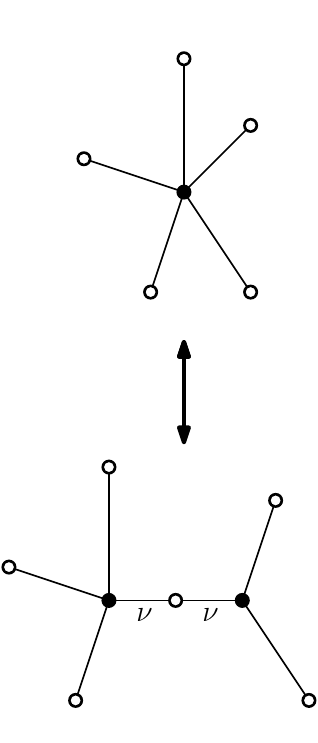}
$\quad\quad\quad\quad\quad$
\includegraphics[width=0.17\textwidth]{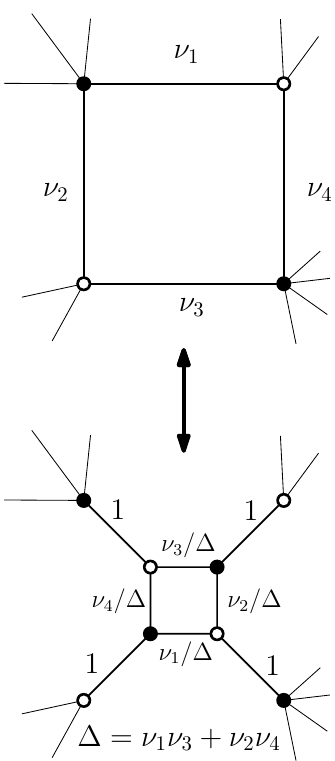}
 \end{center}
\caption{Elementary transformations of weighted bipartite graphs:
(1) parallel edges with weights~$\nu_1, \nu_2$ can be replaced by a single edge with weight~$\nu_1 + \nu_2$; 
(2) contracting a degree~$2$ vertex whose edges have equal weights; 
(3) the spider move, with weight transformation as shown.
}\label{fig:elem}
\end{figure}


There are~$3$ types of elementary transformations of the planar bipartite graph that do not change partition and correlation functions of the dimer model, see Figure~\ref{fig:elem}.
One of the crucial (on the discrete level) properties of t-embeddings is related to elementary transformations of bipartite graphs. 
\begin{proposition}[\cite{KLRR}]\label{prop:elem}
T-embeddings of~$\G^*$ together with their origami maps are preserved under elementary transformations of~$\G$.
\end{proposition}

This result was first proven in~\cite{KLRR}, see also Remark 2.14 and Corollary 2.15 in~\cite{BNR23} for confirmation that the perfectness is also preserved under these transformations.

\subsection{The Aztec diamond and its reduction.}
\begin{figure}
 \begin{center}
 
   \begin{tabular}{c c}
  \begin{minipage}[c]{0.6\textwidth}
\includegraphics[width=1.1\textwidth]{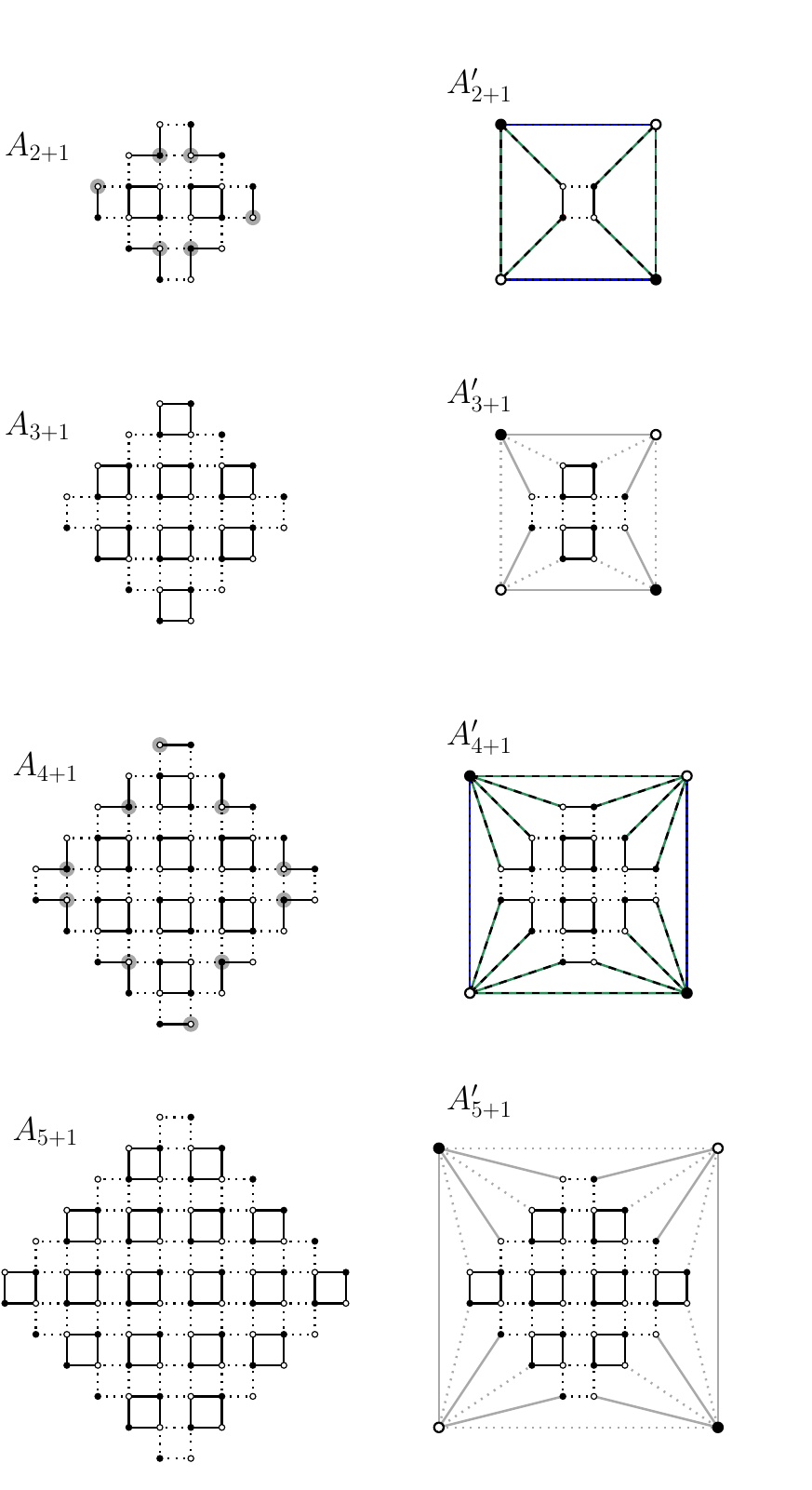}
\end{minipage}
&
\begin{minipage}[c]{0.3\textwidth}

\includegraphics[width=0.12\textwidth]{empty.pdf} 

\includegraphics[width=1\textwidth]{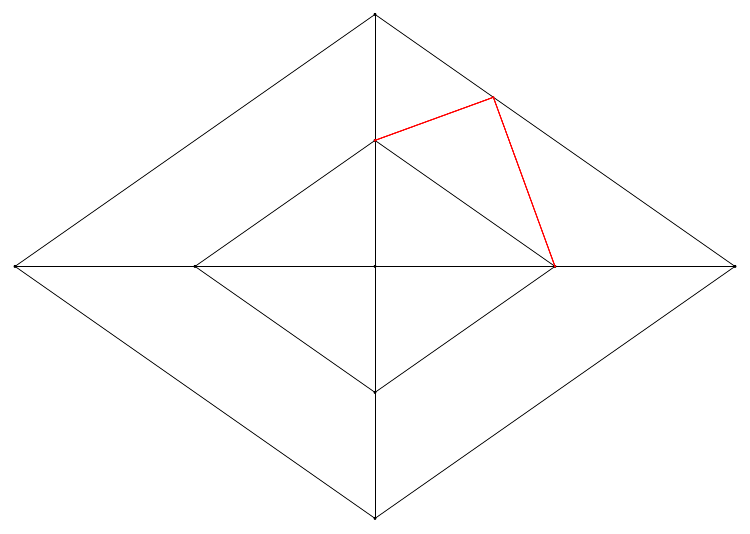} 

\includegraphics[width=0.13\textwidth]{empty.pdf} 

 \includegraphics[width=1\textwidth]{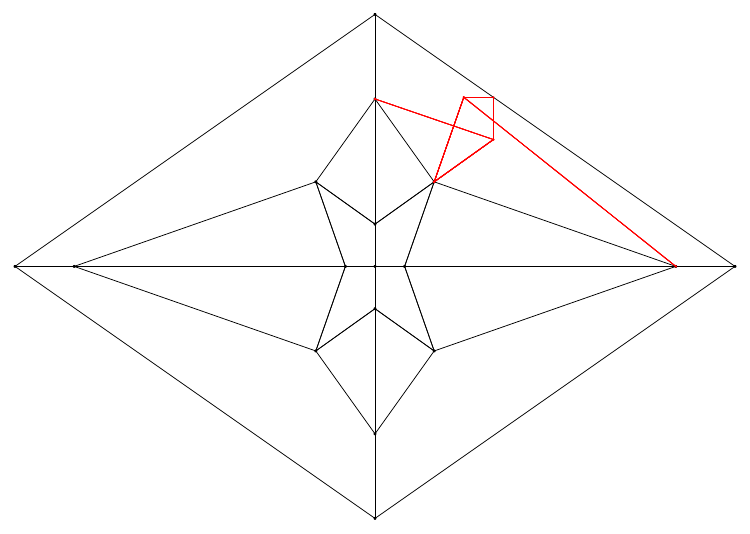} 
 
 \includegraphics[width=0.3\textwidth]{empty.pdf} 

 \includegraphics[width=1\textwidth]{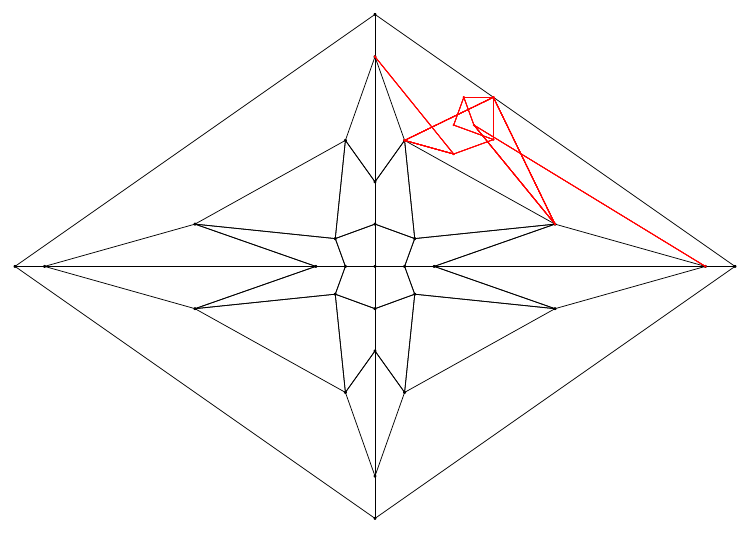}
 
 \includegraphics[width=0.35\textwidth]{empty.pdf} 
 
  \includegraphics[width=1\textwidth]{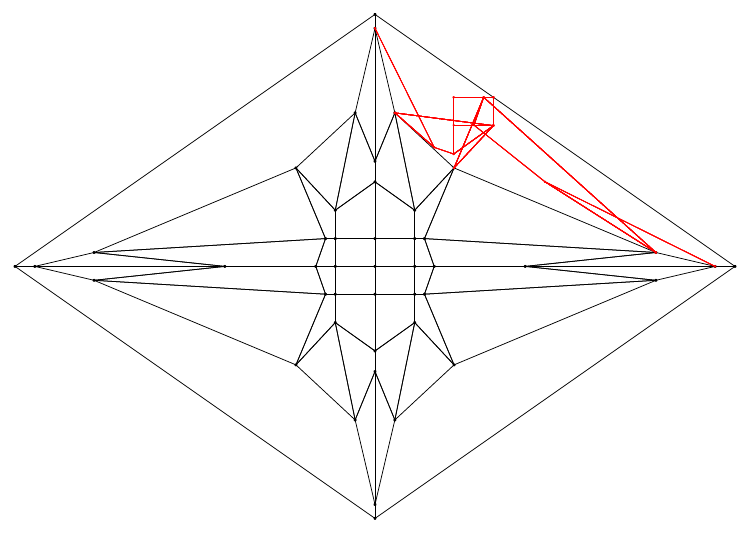} 
  
  \includegraphics[width=0.15\textwidth]{empty.pdf}   
\end{minipage}
\end{tabular}
  \caption{{\bf Left:} Aztec diamonds. 
  The set $V_{\operatorname{gauge}}(A_{n+1})$ is shown in grey.
  For the reduction process multiply edges adjacent to vertices of~$V_{\operatorname{gauge}}(A_{n+1})$ by~$a$. 
  {\bf Middle:} Reduced Aztec diamonds. {\bf Left and middle:} The weights on black dotted edges are~$a$, on black edges are~$1$, on grey dotted edges are~$2a$, on grey edges are~$2$, on green edges are~$a^2$, on blue edges are~$a^3$, on blue edges with black dots are $a^3+a$, and on green and black edges are $a^2+1$. {\bf Right:} perfect t-embedding (black) and its origami map (red) for~$a=0.7$.}\label{fig_period}
 \end{center}
\end{figure}
Consider the square grid~$(\mathbb{Z}+\frac12)^2$. The faces of such a grid can be naturally indexed by pairs~$(j,k)\in\mathbb{Z}^2.$ Let~$n$ be a positive integer. Let us define the \emph{Aztec diamond}~$A_{n+1}$  of size~$n+1$ to be the subset of the faces~$(j,k)$ of the square grid~$(\mathbb{Z}+\frac12)^2$ such that $|j|+|k| \leq n$. 
We also define a $2\times 2$-periodic weight function $\nu$ on edges of the Aztec diamond in the following way. Let $e$ be an edge of the Aztec diamond of size $n$ adjacent to a face $(j,k)$ with $j+k$ odd, then
\begin{equation}\label{two_per_weights}
\nu_e=
\begin{cases}
a, &n=1,2 \mod 4 \text{ and } j \text{ even} \quad \text{or} \quad n=3,0 \mod 4 \text{ and } j \text{ odd}, \\
1, & \text{ otherwise},
\end{cases}
\end{equation}
where $a>0$. We also chose a bipartite coloring of the vertices of the Aztec diamond such that all North-Eastern corners are black, see Figure~\ref{fig_period}. We refer to the Aztec diamond with these weights as the \emph{two-periodic Aztec diamond}.

For $n$ even, let $V_{\operatorname{gauge}}$ be the following set of vertices of the Aztec diamond $A_{n+1}$:
\begin{align*}
\begin{cases}
\text{vertices } \left(j\pm\tfrac12, k\mp\tfrac12\right) \text{ with } j \text{ even and } j+k=\pm n\\
\text{vertices } \left(j\pm\tfrac12, k\pm\tfrac12\right) \text{ with } j \text{ odd and } (\mp j)+(\pm k)=n
\end{cases},
& 
\text{ if } n=4m+2\\
\begin{cases}
\text{vertices } \left(j\mp\tfrac12, k\pm\tfrac12\right) \text{ with } j \text{ even and } j+k=\pm n\\
\text{vertices } \left(j\mp\tfrac12, k\mp\tfrac12\right) \text{ with } j \text{ odd and } (\mp j)+(\pm k)=n
\end{cases},
& 
\text{ if } n=4m.
\end{align*}

Following~\cite{Ch-R}, let us also define the \emph{reduced Aztec diamond}~$A'_{n+1}$  of size~$n+1$. To obtain the reduced Aztec diamond $A'_{n+1}$ from $A_{n+1}$ one should make the following sequence of moves:
\begin{itemize}
\item[$\bullet$] apply a gauge transform (if needed) to modify (only) weights on edges adjacent to faces~$(j,k)$ with~$|j|+|k|=n$. More precisely, for $n$ is even, multiply by $a$ all weights of edges adjacent to vertices of the set~$V_{\operatorname{gauge}}(A_{n+1})$;
\item[$\bullet$] contract black vertices $(j\pm\frac12,k\pm\frac12)$ of $A_{n+1}$ with $j+k=\pm n$;
\item[$\bullet$] contract white vertices $(j\pm\frac12,k\mp\frac12)$ of $A_{n+1}$ with $\begin{cases} j-k=\pm n\\  
\{|j|, |k|\} \neq \{0,n\} \end{cases}$;  
 \item[$\bullet$] merge pairwise all the $4n$ obtained pairs of parallel edges.
\end{itemize}
\begin{remark}
The only difference with reduction process described in~\cite{Ch-R} is the first step: we need this additional step here, since contraction of a vertex of degree~$2$ is possible if and only if the two edge weights are equal to each other.
\end{remark}

Note that the inner vertices of the augmented dual $(A_{n+1})^*$ are in natural correspondence with the inner faces of $A_n$ and can therefore be indexed by $(j,k)\in\mathbb{Z}^2$ with $|j|+|k| < n$. We index the boundary vertex of the augmented dual $(A_{n+1})^*$ adjacent to the dual vertex indexed by~$(n-1,0)$ by~$(n,0)$. Similarly we index the other three boundary vertices of the augmented dual~$(A_{n+1})^*$ by~$(0,n)$, $(-n,0)$ and $(0,-n)$.

\subsection{Recurrence relation for t-embeddings} In this section, we describe the positions of the vertices of a perfect t-embedding of the reduced Aztec diamond obtained from the Aztec diamond with weights~\eqref{two_per_weights} by the reduction procedure described in the previous section.  

The two-periodic Aztec diamond of size $n+1$ 
can be obtained from the Aztec diamond of size $n$ using a sequence of elementary transformations. The same holds for the reduced Aztec diamonds, see Figure~\ref{fig:shuffling}. We use this fact together with Proposition~\ref{prop:elem} to construct perfect t-embeddings~$\T((A_{n+1})^*)$. Denote by~$\T_{n}(j,k)$ the position of the perfect t-embedding of the vertex of~$(A_{n+1})^*$ indexed by~$(j,k)$. The following proposition gives a constructive way to define a sequence of perfect t-embeddings~$\T_n$ of reduced Aztec diamonds with weights~\eqref{two_per_weights}
into a rhombus with diagonals of lengths $2$ and $2a$ using a recurrence relation.

The following proposition is an analogue of a similar result for uniformly weighted Aztec diamond~\cite[Proposition 2.4]{Ch-R} and for uniformly weighted tower graph~\cite[Proposition 4.1]{BNR23}.

\begin{proposition}\label{prop:t_rec} 
The perfect t-embedding~$\T_{n+1}(j,k)$ can be obtained from~$\T_{n}(j,k)$ using the following update rules
\begin{enumerate}
\item $\T_{n+1}(0,\pm(n+1)) =\pm \i a$ 
 and $\T_{n+1}(\pm(n+1),0) = \pm 1$.
\item For $\{|j|, |k|\}=  \{0,n\}$
\begin{equation*}
\T_{n+1}(\pm n,0)=\frac{1}{\alpha_n+1}\Big(\T_n(\pm n,0)+\alpha_n\T_n(\pm(n-1),0)\Big),
\end{equation*}
\begin{equation*}
\T_{n+1}(0,\pm n)=\frac{1}{\alpha_n+1}\Big(\alpha_n\T_n(0,\pm n)+\T_n(0,\pm(n-1))\Big).
\end{equation*}
\item
For $1\leq j\leq n-1$, $|j|+|k|=n$, $\{|j|, |k|\} \neq \{0,n\}$
\[
\T_{n+1}(j,\pm(n-j))
=\frac{1}{\beta_{j,n}+1}\Big(\T_n(j-1,\pm(n-j))+\beta_{j,n}\T_n(j,\pm(n-j-1))\Big).
\]
For $-(n-1)\leq j\leq -1$,  $|j|+|k|=n$
\[
\T_{n+1}(j,\pm(n+j))
=\frac{1}{\beta_{j,n}+1}\Big(\beta_{j,n}\T_n(j,\pm(n+j-1))+\T_n(j+1,\pm(n+j))\Big).
\]
\item For $|j|+|k|<n$ and $j+k+n$ even, $\T_{n+1}(j,k)=\T_n(j,k)$.
\item For $|j|+|k|<n$ and $j+k+n$ odd,
\begin{multline*}
\T_{n+1}(j,k)+\T_n(j,k)=\frac{1}{\gamma_{j,k,n}+1}\Big(\T_{n+1}(j-1,k)+\T_{n+1}(j+1,k) \\
+\gamma_{j,k,n}\left(\T_{n+1}(j,k+1))+\T_{n+1}(j,k-1)\right)\Big).
\end{multline*}
\end{enumerate}
Where the coefficients are given by
\begin{equation*}
\alpha_n=
\begin{cases}
1, & n \text{ odd}, \\
a^2, & n=4m+2,\\
a^{-2}, & n=4m.
\end{cases}
\end{equation*}
\begin{equation*}
\beta_{j,n}=
\begin{cases}
1, & n \text{ odd},\\
a^2, &n=4m \text{ and } j \text{ even} \quad \text{or} \quad n=4m+2 \text{ and } j \text{ odd}, \\
a^{-2}, & n=4m \text{ and } j \text{ odd} \quad \text{or} \quad n=4m+2 \text{ and } j \text{ even}.
\end{cases}
\end{equation*}
\begin{equation*}
\gamma_{j,k,n}=
\begin{cases}
1, & n\text{ even},\\
a^2, & n=4m+3 \text{ and } j,k\text{ even} \quad \text{or}\quad n=4m+1 \text{ and } j,k\text{ odd}, \\
a^{-2}, & n=4m+3 \text{ and } j,k\text{ odd} \quad \text{or}\quad n=4m+1 \text{ and } j,k\text{ even}.
\end{cases}
\end{equation*}
\end{proposition}

\begin{proof} 
The proof mimics the proof of \cite[Proposition 2.4]{Ch-R}. The only difference is that in our setup the edge weights are $(2\times 2)$-periodic instead of uniform ones. 

\begin{figure}
 \begin{center}
\includegraphics[width=0.9\textwidth]{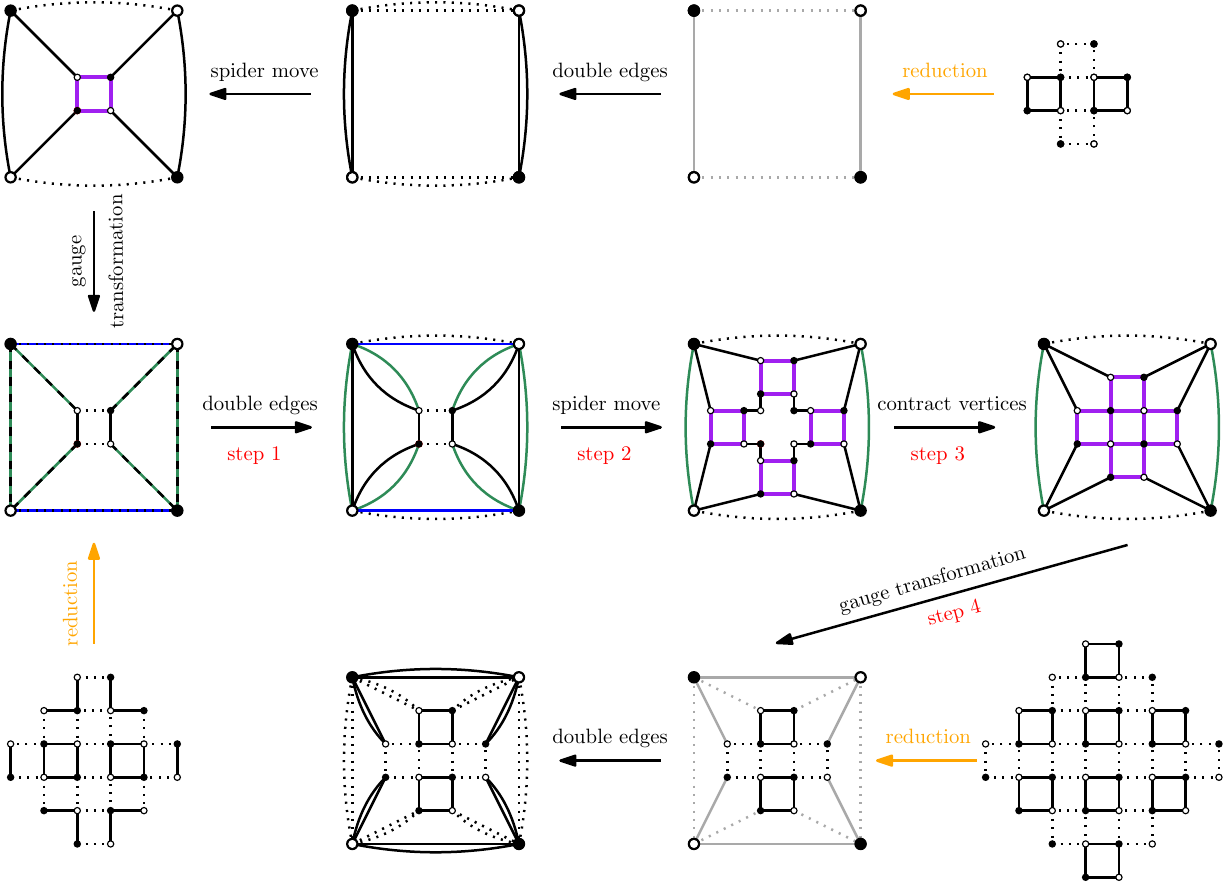}
  \caption{Weights of black dotted edges are~$a$, of black ones are~$1$, of grey dotted edges are~$2a$, and of grey edges are~$2$, of green ones are~$a^2$ and of  blue ones are~$a^3$; weights of blue edges with black dots are~$a^3+a$ and of `green and black' edges are~$a^2+1$. Purple colour corresponds to intermediate weights obtained after spider moves. The steps to get~$A'_{3+1}$ from~$A'_{2+1}$ emphasised in red.}\label{fig:shuffling}
 \end{center}
\end{figure}

Note that given the symmetries of 
the two-periodic Aztec diamond, we can choose the initial data in such a way that $\T_n(0,0)=0$ and all t-embeddings are symmetric with respect to the origin. Then the boundary condition~(1) corresponds to the fact that $A'_{1+1}$ is a square with horizontal edges of weight~$2a$ and vertical ones of weight~$2$ and the geometric weights given by the t-embedding are gauge equivalent to the initial ones.

To get a reduced Aztec diamond~$A'_{n+2}$ from~$A'_{n+1}$ one should make four steps, see Figure~\ref{fig:shuffling}. First, one should split all edges adjacent to at least one of the four boundary vertices into pairs of parallel edges. Let $e$ be an edge adjacent to at least one of boundary vertices.
Let $(j_1,k_1)$ and $(j_2,k_2)$ be two faces adjacent to $e$. Let $e_1, e_2$ be two parallel edges obtained from $e$, such that after splitting, the edge $e_{s}$ is adjacent to the face $(j_{s},k_{s})$ for $s=1,2$. Then the edge weights~$\nu_{e_1}, \nu_{e_2}$ are given by:
\begin{itemize}
\item For $n=4m+1$,
\[
\begin{cases}
\nu_{e_1}=\nu_{e_2}=1 &\text{ if } (j_1,k_1)=(\pm (n-1),0) \text{ and } (j_2,k_2)=(\pm n,0)\\ 
\nu_{e_1}=\nu_{e_2}=a &\text{ if } (j_1,k_1)=(0, \pm (n-1)) \text{ and } (j_2,k_2)=(0, \pm n)\\
\end{cases},
\]
and
\[
\begin{cases}
\nu_{e_1}= \nu_{e_2}=1 &\text{ if } e \text{ is adjacent to a vertex of face } (j,k), \text{ with }  j \text{ even and }|j|+|k|=n-2 \\ 
\nu_{e_1}=\nu_{e_2}=a &\text{ if } e \text{ is adjacent to a vertex of face } (j,k), \text{ with }  j \text{ odd and }|j|+|k|=n-2
\end{cases}.
\]
\item For $n=4m+2$,
\[
\begin{cases}
\nu_{e_1}=1 \text{ and } \nu_{e_2}=a^2 &\text{ if } (j_1,k_1)=(\pm (n-1),0) \text{ and } (j_2,k_2)=(\pm n,0)\\ 
\nu_{e_1}=a^3 \text{ and } \nu_{e_2}=a &\text{ if } (j_1,k_1)=(0, \pm (n-1)) \text{ and } (j_2,k_2)=(0, \pm n)\\
\end{cases},
\]
and
\[
\begin{cases}
\nu_{e_1}=1 \text{ and } \nu_{e_2}=a^2 \quad \text{ if } j \text{ even} \\
\nu_{e_1}=a^2 \text{ and } \nu_{e_2}=1 \quad \text{ if } j \text{ odd} 
\end{cases},
\]
for 
$\begin{cases} (j_1,k_1)=(j\mp1,k)\\
(j_2,k_2)=(j, k\pm1)
\end{cases}$ $\mp j \pm k = n-2$
\quad or $\begin{cases}(j_1,k_1)=(j\pm1,k)\\
(j_2,k_2)=(j, k\pm1)\end{cases}$  $\pm j \pm k = n-2$. 
\end{itemize}
The weights for $n=4m$ (resp. $n=4m+3$) can be obtained by changing the roles of $j$ and $k$ in the case $n=4m+2$ (resp., $n=4m+1$), see Figure~\ref{fig_period}. 
Due to~\cite{KLRR}, this corresponds to adding the points dividing corresponding edges of $\T_n$ in proportions~$[\nu_{e_1}:\nu_{e_2}]$,
i.e. these are update rules~(2) and~(3).

The second step is to apply the spider move at all faces of $A'_{n+1}$ for which $j + k + n$ is odd.
Therefore, (4) reflects the fact that, for~$j + k + n$ even, the inner faces~$(j, k)$ of~$\T_n$ are not destroyed by spider moves and hence the positions of the corresponding dual vertices in the t-embedding~$\T_{n+1}$ remain the same as in~$\T_n$. While (5) corresponds to spider moves at faces for which $j + k + n$ is odd. Here we use the fact that the face weight at the face~$(j,k)$ of~$A'_{n+1}$ is equal to~$\gamma_{j,k,n}$, and the explicit description~\cite[Equation (6)]{KLRR} of the spider move (the \emph{central move}) in terms of the t-embedding.


The third step is the contraction of all vertices of degree~$2$ (note that after a spider move the weights on edges adjacent to such a vertices are all equal to~$1$). 
And the fourth step is a gauge transformation of weights. The existence of such a gauge transformation follows from the fact that the face weights stay the same.
To finish the proof, note that the last two steps 
do not affect positions of the dual vertices in the t-embedding. 
\end{proof}

\begin{remark}
Note that $a=1$ send us to the uniform setup, and the recurrence in~\cite{Ch-R} together with boundary conditions coincide with the one obtained above.
\end{remark}

\begin{remark}
The shuffling algorithm, used in the proof of the above proposition, applied to the Aztec diamond with any edge weights, implies that perfect t-embedding always exists for any weighted Aztec diamond. Moreover, similar to the above proposition, one can write a recurrence formula for a perfect t-embedding~$\T_n((A'_n)^*)$ of the form
\begin{multline*}
\T_{n+1}(j,k)+\T_n(j,k)=\frac{1}{c_{j,k,n}+1}\Big(\T_{n+1}(j-1,k-1)+\T_{n+1}(j+1,k+1) \\
+c_{j,k,n}\left(\T_{n+1}(j-1,k+1))+\T_{n+1}(j+1,k-1)\right)\Big).
\end{multline*}
However, identifying the coefficients~$c_{j,k,n}$ even in the case of general doubly periodic edge weights is a non-trivial task.
While the recurrence relations appear difficult to use for theoretical analysis in broad generality -- beyond the connection to the octahedron equation discussed in this paper -- they remain valuable for simulations of finite-sized t-embeddings.
\end{remark}

Let~$\Or_{n}(j,k)$ be the origami map of~$(A'_{n+1})^*$, where the root white face is the one adjacent to~$\T_n(n,0)$ and~$\T_n(0,n)$. Then, similar to~\cite{Ch-R}, the following holds.

\begin{proposition}\label{prop:o_rec}
For $n \geq 1$ the origami map $\Or_{n+1}$ can be constructed from $\Or_n$ using the same update rules (2)--(5) as in Proposition~\ref{prop:t_rec}, with boundary conditions 
\[ \Or_n(n,0) = \Or_n(-n,0) = 1\quad\text{ and } \quad\Or_n(0,n)=\Or_n(0,-n) = \i a.\]
\end{proposition}

The update rules (2), (3) and (5) from Proposition~\ref{prop:t_rec} can be rewritten in a more universal way. 
Let $\Lambda = \{(j,k,n)\in \ZZ^2\times \ZZ;j+k+n \text{ odd}\}.$
Given $(b_0,b_E,b_N,b_W,b_S) \in \CC^5 $ the following conditions define the function $f:\Lambda \to \CC$ uniquely. 
\begin{enumerate}
\item For all $(j,k)\in\mathbb{Z}^2$,
\begin{equation} \label{bc}
f(j,k,0)=0 \quad\text{ and }\quad f(j,k,-1)=0;
\end{equation}
\item If $j+k+n$ is even
\begin{equation} \label{eq:rec}
\begin{split}
&f(j,k,n+1)+f(j,k,n-1)\\
 &- \frac{f(j-1,k,n)+f(j+1,k,n)+\tilde\gamma_{j,k,n}\big(f(j,k+1,n)+f(j,k-1,n)\big)}{\tilde\gamma_{j,k,n}+1}\\
&\quad=
\delta_{j,n}\delta_{k,0}\frac{1\cdot b_E}{\tilde\gamma_{j,k,n}+1}+
\delta_{j,-n}\delta_{k,0}\frac{1\cdot b_W}{\tilde\gamma_{j,k,n}+1}+
\delta_{j,0}\delta_{k,n}\frac{\tilde\gamma_{j,k,n} \cdot b_N}{\tilde\gamma_{j,k,n}+1}+
\delta_{j,0}\delta_{k,-n}\frac{\tilde\gamma_{j,k,n}\cdot b_S}{\tilde\gamma_{j,k,n}+1}\\
&\quad\quad+\delta_{j,0}\delta_{k,0}\delta_{n,1}\cdot b_0,
\end{split}
\end{equation}
\end{enumerate}
where the coefficients $\tilde\gamma_{j,k,n}$ are given by
\begin{equation}\label{eq:tilde_gamma}
\tilde\gamma_{j,k,n}=
\begin{cases}
1, & n\text{ odd},\\
a^2, & n=4m \text{ and } j,k\text{ even} \quad \text{or}\quad n=4m+2 \text{ and } j,k\text{ odd}, \\
a^{-2}, & n=4m \text{ and } j,k\text{ odd} \quad \text{or}\quad n=4m+2 \text{ and } j,k\text{ even}.
\end{cases}
\end{equation}

Note that $f(j, k, n) = 0$ if $n\leq 0$  or 
$\begin{cases} |j| + |k|\geq n\\
\{|j|, |k|\} \neq \{0,n\}
\end{cases}$.

\begin{remark}
Note that $\tilde\gamma_{j,k,n}=\gamma_{j,k,n-1}$, $\tilde\gamma_{j,n-j,n}=\beta_{j,n}$ and $\tilde\gamma_{0,\pm n,n}=\tilde\gamma_{\pm n,0,n}=\alpha^{-1}_{n}$. Therefore, $\T_n(j,k)$ can be seen as the solution~$f(j,k,n)$ of the above system with the boundary conditions~$(0,1,\i a, -1, -\i a)$. To see this we use that~$\T_{n+1}(j,k)=\T_n(j,k)$ for~$j+k+n$ even. Similarly, the origami map corresponds to the boundary conditions~$(0,1,\i a, 1, \i a)$. 
\end{remark}

\begin{remark}
In the `uniform case', i.e. for $a=1$, equation~\eqref{eq:rec} has the following form
\begin{equation} \label{eq:rec-unif}
\begin{split}
f&(j,k,n+1)+f(j,k,n-1)\\
 &- \frac12\big(f(j-1,k,n)+f(j+1,k,n)+f(j,k+1,n)+f(j,k-1,n)\big)\\
&\quad=\delta_{j,0}\delta_{k,0}\delta_{n,0}\cdot b_0 +
\frac12\big(\delta_{j,n}\delta_{k,0}\cdot b_E+
\delta_{j,-n}\delta_{k,0}\cdot b_W+
\delta_{j,0}\delta_{k,n}\cdot b_N+
\delta_{j,0}\delta_{k,-n}\cdot b_S\big).
\end{split}
\end{equation}
\end{remark}

Let $f_0$ be the fundamental solution of the system~\eqref{eq:rec} with initial conditions~\eqref{bc}, i.e. the solution with the boundary condition \[(b_0,b_E,b_N,b_W,b_S)=(1,0,0,0,0).\]  Let $f_E$, $f_N$, $f_W$ and $f_S$  be the solutions with boundary condition $(0,1,0,0,0)$ and so on. Then the solution to the equation with boundary conditions $(b_0,b_E,b_N,b_W,b_S)$ is given by
\begin{equation}
f=b_0f_0+b_Ef_E+b_Nf_N+b_Wf_W+b_Sf_S.
\end{equation} 
In particular, for $(j,k,n)\in \Lambda_+$
\begin{equation}\label{eq:t_n_f_e}
\T_n(j,k)=f_E(j,k,n)+\i{ a} f_N(j,k,n)-f_W(j,k,n)-\i { a} f_S(j,k,n)
\end{equation}
and
\begin{equation}\label{eq:o_n_f_e}
\Or_n(j,k)=f_E(j,k,n)+\i{ a} f_N(j,k,n)+f_W(j,k,n)+\i { a} f_S(j,k,n).
\end{equation}

The functions~$f_E$,~$f_W$,~$f_N$, and~$f_S$ can be expressed in terms of the fundamental solution~$f_0$. In the uniform setting, this was observed in~\cite{Ch-R}.
\begin{lemma}[\cite{Ch-R}]\label{lem:f_e_unif} In the uniform case, the function $f_E$ can be expressed via fundamental solution~$f_0$ in the following way
\begin{equation}\label{eq:f_E_unif}
f_E(j,k,n)=\tfrac12\sum_{s=0}^{n+1} f_{0}(j-s,k,n-s).
\end{equation}
Similarly for $f_W$, $f_N$ and $f_S$ one has
\[f_W(j,k,n)=\tfrac12\sum_{s=0}^{n+1} f_{0}(j+s,k,n-s),\]
\[f_N(j,k,n)=\tfrac12\sum_{s=0}^{n+1} f_{0}(j,k-s,n-s) \quad\text{and}\quad 
f_S(j,k,n)=\tfrac12\sum_{s=0}^{n+1} f_{0}(j,k+s,n-s).
\]
\end{lemma}
\begin{proof} 
Note that the RHS of~\eqref{eq:f_E_unif} satisfies Equation~\eqref{eq:rec-unif} for $(b_0,b_E,b_N,b_W,b_S)=(0,1,0,0,0)$ and vanishes at $n=0$ and $n=-1$. This conditions identify the function $f_E$ uniquely. The functions $f_N$, $f_W$, and $f_S$ are identified similarly.
\end{proof}

To get a similar representation for functions $f_E$, $f_W$, $f_N$ and $f_S$ in the $2\times2$-periodic setup we need to introduce four different `shifted' fundamental solutions. For $\epsilon, \eta \in \{0,1\}$, let us define four functions~$f_{(\epsilon,\eta)}:\Lambda\to \CC$  as a solutions of recurrence relation~\eqref{eq:rec} with a modified RHS. More precisely, the functions $f_{(\epsilon,\eta)}$ satisfy the initial conditions~\eqref{bc} and the following recurrence relation
\begin{equation} \label{eq:rec_fund}
\begin{split}
&f_{(\epsilon,\eta)}(j,k,n+1)+f_{(\epsilon,\eta)}(j,k,n-1)\\
 &- \frac{f_{(\epsilon,\eta)}(j-1,k,n)+f_{(\epsilon,\eta)}(j+1,k,n)+
 \tilde\gamma_{j,k,n}\big(f_{(\epsilon,\eta)}(j,k+1,n)+f_{(\epsilon,\eta)}(j,k-1,n)\big)}{\tilde\gamma_{j,k,n}+1}\\
&\quad\quad\quad\quad\quad\quad=
\begin{cases}
\delta_{j,0}\delta_{k,0}\delta_{n,0} &\text{ if } (\epsilon,\eta)=(0,0)\\
\delta_{j,1}\delta_{k,1}\delta_{n,0} &\text{ if } (\epsilon,\eta)=(1,1)\\
-\delta_{j,0}\delta_{k,1}\delta_{n,1}&\text{ if } (\epsilon,\eta)=(0,1)\\
-\delta_{j,1}\delta_{k,0}\delta_{n,1}&\text{ if } (\epsilon,\eta)=(1,0)
\end{cases}.
\end{split}
\end{equation}
Note that~$f_{(0,0)}$ is the fundamental solution~$f_0$ defined above, nevertheless we change the notation to have a similar notation for all four functions~$f_{(\epsilon,\eta)}$. 

\begin{lemma}\label{lem:f_E_repr}
Let $m=\lfloor\frac{n}{4}\rfloor$ and $f_{(\epsilon,\eta)}$ as defined above. Then for $n \geq -1$ the following holds
\begin{equation}\label{eq:f_E}
\begin{split}
f_E(j,k,n)=
&\tfrac{1}{1+a^{-2}}\sum_{s=0}^{m+1}f_{(0,0)} (j-4s, \,k,\, n-4s)\\
&\quad+\tfrac{1}{1+a^{2}}\sum_{s=0}^{m+1}f_{(1,1)} (j-4s-1,\, k+1, \, n-4s-2)\\
&\quad\quad-\tfrac{1}{2}\sum_{s=0}^{m+1}f_{(0,1)} (j-4s-3, \, k+1,\, n-4s-2)\\
&\quad\quad\quad-\tfrac{1}{2}\sum_{s=0}^{m+1}f_{(1,0)} (j-4s, \, k, \, n-4s);
\end{split}
\end{equation}

\begin{equation}\label{eq:f_W}
\begin{split}
f_W(j,k,n)=
&\tfrac{1}{1+a^{-2}}\sum_{s=0}^{m+1}f_{(0,0)} (j+4s, \, k, \, n-4s)\\
&\quad+\tfrac{1}{1+a^{2}}\sum_{s=0}^{m+1}f_{(1,1)} (j+4s+3, \, k+1, \, n-4s-2)\\
&\quad\quad-\tfrac{1}{2}\sum_{s=0}^{m+1}f_{(0,1)} (j+4s+3, \, k+1, \, n-4s-2)\\
&\quad\quad\quad-\tfrac{1}{2}\sum_{s=0}^{m+1}f_{(1,0)} (j+4s, \, k, \, n-4s);
\end{split}
\end{equation}

\begin{equation}\label{eq:f_S}
\begin{split}
f_S(j,k,n)=
&\tfrac{1}{1+a^{2}}\sum_{s=0}^{m+1}f_{(0,0)} (j, \, k+4s, \, n-4s)\\
&\quad+\tfrac{1}{1+a^{-2}}\sum_{s=0}^{m+1}f_{(1,1)} (j+1, \, k+4s+3, \, n-4s-2)\\
&\quad\quad-\tfrac{1}{2}\sum_{s=0}^{m+1}f_{(0,1)} (j+1, \, k+4s+3, \, n-4s-2)\\
&\quad\quad\quad-\tfrac{1}{2}\sum_{s=0}^{m+1}f_{(1,0)} (j-4s, \, k, \, n-4s);
\end{split}
\end{equation}

\begin{equation}\label{eq:f_N}
\begin{split}
f_N(j,k,n)=
&\tfrac{1}{1+a^{2}}\sum_{s=0}^{m+1}f_{ (0,0)} (j, \, k-4s, \, n-4s)\\
&\quad+\tfrac{1}{1+a^{-2}}\sum_{s=0}^{m+1}f_{(1,1)} (j+1, \, k-4s-1, \, n-4s-2)\\
&\quad\quad-\tfrac{1}{2}\sum_{s=0}^{m+1}f_{ (0,1)} (j+1, \, k-4s-3, \, n-4s-2)\\
&\quad\quad\quad-\tfrac{1}{2}\sum_{s=0}^{m+1}f_{ (1,0)}(j, \, k-4s, \, n-4s).
\end{split}
\end{equation}

\end{lemma}
\begin{proof} Let us show that~\eqref{eq:f_E} holds. The representations of  $f_W$, $f_S$ and $f_N$ can be checked similarly.
Note that the LHS of~\eqref{eq:rec_fund} coincide with the LHS of~\eqref{eq:rec}. Note also, that 
\[\tilde\gamma_{j, \, k, \, n} = \tilde\gamma_{j+4s_1, \, k+4s_2, \, n+4s_3} 
\quad \text{ and } \quad
\tilde\gamma_{j, \, k, \, n} = \tilde\gamma_{j+2s_1+1, \, k+2s_2+1, \, n+4s_3+2}, \]
for all $s_1, s_2, s_3 \in \mathbb{Z}$.
Therefore, it remains to check that the RHS of~\eqref{eq:f_E} satisfies Equation~\eqref{eq:rec} for $(b_0,b_E,b_N,b_W,b_S)=(0,1,0,0,0)$ and vanishes at $n=0$ and $n=-1$, since this conditions identify the function $f_E$ uniquely. For $n=-1$ (resp., $n=0$) all sums in the RHS of~\eqref{eq:f_E} has exactly one (resp., two) terms, and all these terms have a non-positive last argument, hence the RHS of~\eqref{eq:f_E} vanishes at $n=0$ and $n=-1$. Finally, one can check, that the RHS of~\eqref{eq:rec_fund} implies that the RHS of~\eqref{eq:f_E} satisfies Equation~\eqref{eq:rec} for $(b_0,b_E,b_N,b_W,b_S)=(0,1,0,0,0)$.
\end{proof}

\section{T-embeddings and the octahedron equation}\label{sec:t-emb_octa}
In this section we show a connection between t-embeddings and the octahedron equation. More precisely, we show that the fundamental solution introduced in previous section has a probabilistic interpretation and the t-embedding itself can be written as a sum of density functions introduced in~\cite{DF-SG}. 

\subsection{The octahedron equation and dimer models}\label{sec:octa}
Let us first recall the setup of paper~\cite{DF-SG}. 

For $j, k, n \in\mathbb{Z}$, $n\geq 1$, $j+k+n$ odd, let $T^{\operatorname{oct}}_{j,k,n}\in\mathbb{R}$ be the solution of the \emph{octahedron recurrence} 
\begin{equation}\label{oct_rec}
T^{\operatorname{oct}}_{j,k,n+1}T^{\operatorname{oct}}_{j,k,n-1}=
T^{\operatorname{oct}}_{j+1,k,n}T^{\operatorname{oct}}_{j-1,k,n}+
T^{\operatorname{oct}}_{j,k+1,n}T^{\operatorname{oct}}_{j,k-1,n},
\end{equation}
with  the initial condition
 \begin{equation}\label{oct_rec_bc}
T^{\operatorname{oct}}_{j, k, n_{jk}} = t_{j,k},
\end{equation}
where $ t_{j,k}\in\mathbb{R}$ and $n_{jk}=j+k+1 \mod 2$. The solution $T^{\operatorname{oct}}_{j,k,n}$ is known to be related to the partition function of the dimer model on the Aztec diamond. 

Let~$A_{j,k,n}$ be the Aztec diamond of size~$n-1$ centered at~$(j,k)\in\mathbb{Z}^2$, i.e. it has faces~$(\epsilon, \eta)$ such that~$|\epsilon-j|+|\eta-k|<n-1$. Let us also define the closure of $A_{j,k,n}$ by
\[\bar A_{j,k,n}=\{(\epsilon, \eta) \text{ such that } |\epsilon-j|+|\eta-k|\leq n-1\}.\]
Given~$t_{\epsilon, \eta}\in\mathbb{R}_{+}$ with~$\epsilon, \eta\in\mathbb{Z}$, let us define a weight function on edges of the Aztec diamond~$A_{j,k,n}$ by 
\begin{equation}\label{edge_weights}
\hat\nu_e=\frac{1}{t_{\epsilon_e,\eta_e}t_{\epsilon'_e,\eta'_e}},
\end{equation}
where $(\epsilon_e,\eta_e), (\epsilon'_e,\eta'_e) \in \bar A_{j,k,n}$ are faces adjacent to the edge $e$. 
Then, following~\cite{DF-SG}, the weight of each dimer configuration $M$ can be written as 
\[\prod_{e\in M} \hat\nu_e
=\prod_{(\epsilon, \eta)\in \bar A_{j,k,n}} (t_{\epsilon, \eta})^{-\mathcal D_{\epsilon, \eta}} 
= \prod_{(\epsilon, \eta)\in \bar A_{j,k,n}} \frac{1}{t_{\epsilon, \eta}}\prod_{(\epsilon, \eta)\in \bar A_{j,k,n}} (t_{\epsilon, \eta})^{1-\mathcal D_{\epsilon, \eta}} ,\]
where $\mathcal D_{\epsilon, \eta}$ is the number of dimers around the face $(\epsilon, \eta)$ in the dimer configuration $M$ and $\mathcal M$ is the set of all possible dimer configurations. 
Similarly, the partition function of the dimer model on Aztec diamond~$A_{j,k,n}$ can be written as
\[Z_{j,k,n}:=\sum_{M\in\mathcal{M}} \prod_{e\in M} \hat\nu_e=
\prod_{(\epsilon, \eta)\in \bar A_{j,k,n}} \frac{1}{t_{\epsilon, \eta}}
\left(\sum_{M\in\mathcal{M}}
\prod_{(\epsilon, \eta)\in \bar A_{j,k,n}} (t_{\epsilon, \eta})^{1-\mathcal D_{\epsilon, \eta}}
\right).\]

\begin{remark} The weight of a dimer configuration defined in~\cite{DF-SG} differ by $\left(\prod\limits_{(\epsilon, \eta)\in \bar A_{j,k,n}} {t_{\epsilon, \eta}}\right)$ from the one defined above. However, the probability measure stays the same.
\end{remark}

The following proposition is a version of~\cite[Theorem 2.2]{DF-SG} with the extra product of $t_{\epsilon, \eta}$'s taken into account.

\begin{proposition}[\cite{Sp, DF}]\label{T=Z}
The solution $T^{\operatorname{oct}}_{j,k,n}$ of the octahedron recurrence~\eqref{oct_rec} with the initial condition ~\eqref{oct_rec_bc} and the partition function~$Z_{j,k,n}$ of the dimer model on the Aztec diamond~$A_{j,k,n}$ with edge weights defined by~\eqref{edge_weights} satisfy
\[T^{\operatorname{oct}}_{j,k,n}=
Z_{j,k,n} \prod_{(\epsilon, \eta)\in \bar A_{j,k,n}} {t_{\epsilon, \eta}}.
\]
\end{proposition}

Following~\cite{DF-SG}, for $(\epsilon, \eta)\in \mathbb{Z}^2, n\geq 0$ let us define the density function $\rho^{(\epsilon, \eta)}_{j,k,n}$ by
\begin{equation}\label{eq:def_rho}
\rho^{(\epsilon, \eta)}_{j,k,n}=t_{\epsilon, \eta}
\partial_{t_{\epsilon, \eta}}
\log{T^{\operatorname{oct}}_{j,k,n}}|_{t_{a,b}=t^*_{a,b}}.
\end{equation}

The reason why it is called the density function in~\cite{DF-SG} is clarified in the following remark. 

\begin{remark}\label{rem:density_function}
Given~\eqref{eq:def_rho} and Proposition~\ref{T=Z}, one can easily check that for any $(\epsilon, \eta)\in\bar{A}_{j,k,n}$
\begin{equation}\label{eq:rho_prob}
\text{ for } n\geq 1 \quad\quad
\rho^{(\epsilon, \eta)}_{j,k,n}=\EE_{j,k,n}\left[1-\mathcal D_{\epsilon, \eta}\right],
\end{equation}
where $\mathcal D_{\epsilon, \eta}$ is a number of dimers around face $(\epsilon, \eta)$ in a random dimer configuration of the Aztec diamond~$A_{j,k,n}$. Indeed, note that 
\begin{multline*}
t_{\epsilon, \eta}
\partial_{t_{\epsilon, \eta}}
\Log{T^{\operatorname{oct}}_{j,k,n}}\\
=t_{\epsilon, \eta}\frac{\partial_{t_{\epsilon, \eta}}T^{\operatorname{oct}}_{j,k,n}}{T^{\operatorname{oct}}_{j,k,n}} 
=
t_{\epsilon, \eta}\frac{\sum\limits_{\operatorname{conf}}\partial_{t_{\epsilon, \eta}}\prod\limits_{(x, y)} (t_{x, y})^{1-\mathcal D_{x, y}}}{Z_{j,k,n} \prod\limits_{(x, y)
} {t_{x, y}}}=
\sum\limits_{M\in \mathcal M}\left(1-\mathcal D_{\epsilon, \eta}\right)\frac{ \prod\limits_{(x, y)} \frac{1}{t_{x,y}}\prod\limits_{(x, y)} (t_{x, y})^{1-\mathcal D_{x, y}} }{Z_{j, k ,n}},
\end{multline*}
where the product is over all $(x, y)\in \bar A_{j,k,n}$. 
\end{remark}

Note that the contribution of an extra product of~$t_{\epsilon, \eta}$'s completely disappeared here and our density function coincides with the one defined in~\cite{DF-SG}.

Due to~\cite{DF-SG}, for $n\geq1$, the density function $\rho^{(\epsilon, \eta)}_{j,k,n}$ satisfies the following recurrence relation:
\begin{equation}\label{eq:rho_rec}
\rho^{(\epsilon, \eta)}_{j,k,n+1}+\rho^{(\epsilon, \eta)}_{j,k,n-1}=
L_{j,k,n}\left(\rho^{(\epsilon, \eta)}_{j+1,k,n}+\rho^{(\epsilon, \eta)}_{j-1,k,n}\right)
+R_{j,k,n}\left(\rho^{(\epsilon, \eta)}_{j,k+1,n}+\rho^{(\epsilon, \eta)}_{j,k-1,n}\right),
\end{equation}
where 
\begin{equation}\label{eq:coefficients_R_L}
R_{j,k,n}=\frac{T^{\operatorname{oct}}_{j+1,k,n}T^{\operatorname{oct}}_{j-1,k,n}}{T^{\operatorname{oct}}_{j,k,n+1}T^{\operatorname{oct}}_{j,k,n-1}} \quad\text{and}\quad L_{j,k,n}=1-R_{j,k,n}.
\end{equation}

\subsection{Uniform initial conditions}
In the uniform case, the initial data is given by $t_{j,k}=1$. Indeed, in this case, the solution of the T-system $T^{\operatorname{oct}}_{j,k,n}=2^{n(n-1)/2}$ coincides with the partition function of the uniform dimer model of an Aztec diamond of size $n$. The recurrence relation in the uniform case is the following
\begin{equation}\label{eq:rho_rec_unif}
\rho^{(\epsilon, \eta)}_{j,k,n+1}+\rho^{(\epsilon, \eta)}_{j,k,n-1}=
\frac12\left(\rho^{(\epsilon, \eta)}_{j+1,k,n}+\rho^{(\epsilon, \eta)}_{j-1,k,n}
+\rho^{(\epsilon, \eta)}_{j,k+1,n}+\rho^{(\epsilon, \eta)}_{j,k-1,n}\right).
\end{equation}

\begin{proposition}\label{prop:fund_prob}
The density function~$\rho^{(0, 0)}_{j,k,n}$ extended to~$n=-1$ by~$\rho^{(0, 0)}_{j,k,-1}=0$ is a fundamental solution of the recurrence equation~\eqref{eq:rec-unif}.
\end{proposition}

\begin{proof}
Using a probabilistic meaning of the density function given in~\eqref{eq:rho_prob} one can easily check that~$\rho^{(0, 0)}_{j,k,0}=0$ and~$\rho^{(0, 0)}_{j,k,1}=\delta_{j,0}\delta_{k,0}$.  Now the initial data $\rho^{(0, 0)}_{j,k,-1}=0$ together with~\eqref{eq:rho_rec_unif} implies that~$\rho^{(0, 0)}_{j,k,n}$ is a fundamental solution of the recurrence equation~\eqref{eq:rec-unif}.
\end{proof}

\begin{corollary} The previous proposition implies that~\eqref{eq:f_E_unif} can be rewritten as 
\[
f_E(j,k,n)=\tfrac12\sum_{s=0}^{n-1} \rho^{(0, 0)}_{j-s,k,n-s}
\]
and similarly for $f_N, f_W$ and $f_S$.
\end{corollary}

\subsection{The~$(2\times2)$-periodic initial condition}
 Following~\cite{DF-SG} let us define 
\begin{equation}\label{tjk}
t_{j,k}=
\begin{cases}
a_{\operatorname{oct}} & \text{ if } j=0,\, k=0\mod 2\\
b_{\operatorname{oct}} & \text{ if } j=1,\, k=1\mod 2\\
c_{\operatorname{oct}} & \text{ if } j=0,\, k=1\mod 2\\
d_{\operatorname{oct}} & \text{ if } j=1,\, k=0\mod 2,
\end{cases}
\end{equation} 
where $a_{\operatorname{oct}}, b_{\operatorname{oct}}, c_{\operatorname{oct}}, d_{\operatorname{oct}}\in\mathbb{R}.$
Note that $t_{j+2,k}=t_{j,k}$ and $t_{j,k+2}=t_{j,k}$. 

In this case,~\cite[Lemma 3.1]{DF-SG} says that~$T^{\operatorname{oct}}_{j,k,n}$ can be rewritten in the following form.
\begin{lemma}[\cite{DF-SG}]\label{lem:DF-SG}
Let~$T^{\operatorname{oct}}_{j,k,n}$ be the solution to~\eqref{oct_rec} with initial conditions~\eqref{tjk}, then  
\[T^{\operatorname{oct}}_{j,k,n}=
\left(\frac{a^2_{\operatorname{oct}}+b^2_{\operatorname{oct}}}{c_{\operatorname{oct}}d_{\operatorname{oct}}}\right)
^{\lfloor\frac{n}{2}\rfloor\lfloor\frac{n+1}{2}\rfloor}
\left(\frac{c^2_{\operatorname{oct}}+d^2_{\operatorname{oct}}}{a_{\operatorname{oct}}b_{\operatorname{oct}}}\right)
^{\lfloor\frac{n-1}{2}\rfloor\lfloor\frac{n}{2}\rfloor}
\times
\begin{cases}
t_{j,k} &\text{ if } n=0,1 \mod 4\\
t_{j+1,k+1} &\text{ if } n=2,3 \mod 4
\end{cases}.
\]
The coefficients~$L_{j,k,n}$ from \eqref{eq:coefficients_R_L} are given by
\[
L_{j,k,n}=
\left(\frac{a^2_{\operatorname{oct}}+b^2_{\operatorname{oct}}}{c_{\operatorname{oct}}d_{\operatorname{oct}}}\right)
^{\lfloor\frac{n}{2}\rfloor-\lfloor\frac{n+1}{2}\rfloor}
\left(\frac{c^2_{\operatorname{oct}}+d^2_{\operatorname{oct}}}{a_{\operatorname{oct}}b_{\operatorname{oct}}}\right)
^{\lfloor\frac{n-1}{2}\rfloor-\lfloor\frac{n}{2}\rfloor}
\times\tau_{j,k,n},
\]
where
\[
\tau_{j,k,n}=
\begin{cases}
(t_{j+1,k}t_{j-1,k}) / (t_{j,k}t_{j+1,k+1}) &\text{ if } n=0,1 \mod 4\\
(t_{j+2,k+1}t_{j,k+1}) / (t_{j,k}t_{j+1,k+1}) &\text{ if } n=2,3 \mod 4.
\end{cases}
\]
\end{lemma}

The following lemma shows that in $2\times2$-periodic case for some special choice of the parameters~$a_{\operatorname{oct}}, b_{\operatorname{oct}}, c_{\operatorname{oct}}, d_{\operatorname{oct}}$
the recurrence relation~\eqref{eq:rho_rec} coincide with~\eqref{eq:rec}.

\begin{lemma} 
Let $a$ be a positive real number and~$\tilde \gamma_{i,j,k}$ be given by~\eqref{eq:tilde_gamma}. Assume also that the real numbers $a_{\operatorname{oct}}, b_{\operatorname{oct}}, c_{\operatorname{oct}}, d_{\operatorname{oct}}$ satisfy 
\[a_{\operatorname{oct}}=b_{\operatorname{oct}} \quad \text{ and }\quad c_{\operatorname{oct}}=ad_{\operatorname{oct}}.\] 
Then in the setup of Lemma~\ref{lem:DF-SG} one has $L_{i,j,k}=\frac{1}{1+\tilde\gamma_{i,j,k}}$.
\end{lemma}
\begin{proof}
Note that 
\begin{equation*}
L_{i,j,k} =
\begin{cases}
a^2_{\operatorname{oct}}/(a^2_{\operatorname{oct}}+b^2_{\operatorname{oct}})&\text{ or }\quad
b^2_{\operatorname{oct}}/(a^2_{\operatorname{oct}}+b^2_{\operatorname{oct}}),  \qquad k\text{ odd},\\
d^2_{\operatorname{oct}}/(c^2_{\operatorname{oct}}+d^2_{\operatorname{oct}}), & \quad k=4m \text{ and } i,j\text{ even} \quad \text{or}\quad k=4m+2 \text{ and } i,j\text{ odd}, \\
c^2_{\operatorname{oct}}/(c^2_{\operatorname{oct}}+d^2_{\operatorname{oct}}), &\quad k=4m \text{ and } i,j\text{ odd} \quad \text{or}\quad k=4m+2 \text{ and } i,j\text{ even}.
\end{cases}
\end{equation*}
Therefore $L_{i,j,k}=\frac{1}{1+\tilde\gamma_{i,j,k}}$ for $a_{\operatorname{oct}}=b_{\operatorname{oct}}$ and $c_{\operatorname{oct}}=a d_{\operatorname{oct}}$.
\end{proof}

\begin{remark}\label{rem:weights}
Let $t_{j,k}$ be given by~\eqref{tjk} with $a_{\operatorname{oct}}=b_{\operatorname{oct}}=c_{\operatorname{oct}}=1$ and $d_{\operatorname{oct}}=1/a$. 
Let~$e$ be an edge adjacent to a face~$(j,k)$ with~$j+k$ odd, then
\begin{equation}\label{two_per_weights_oct}
\hat\nu_e=
\begin{cases}
a,  &j \text{ odd}, \\
1, &j \text{ even},
\end{cases}
\end{equation}
where~$\nu_e$ is defined by~\eqref{edge_weights}. Note that this weights do not depend on the size of the Aztec diamond, however the weights~$\nu_e$ defined in~\eqref{two_per_weights} do.
\end{remark}



Let us now formulate an analogue of Proposition~\ref{prop:fund_prob} for the $(2\times2)$-periodic case. Recall the definitions of~$f_{(\epsilon, \eta)}$ and~$\rho^{(\epsilon, \eta)}_{j,k,n}$ in~\eqref{eq:rec_fund} and~\eqref{eq:def_rho}. 

\begin{proposition}\label{prop:shift_fund_prob}
In the setup described above with~~$a_{\operatorname{oct}}=b_{\operatorname{oct}}$ and~$c_{\operatorname{oct}}=ad_{\operatorname{oct}}$ the following holds
\begin{enumerate}
\item for $j+k+n$ odd, $\eta \in \{0,1\}$, $n\geq0$
\[f_{(\eta,\eta)}(j, k, n)=\rho^{(\eta, \eta)}_{j,k,n};\]
\item for $j+k+n$ odd, $\epsilon,\eta \in \{0,1\}$, $\epsilon\neq\eta$, $n\geq1$
\[f_{(\epsilon,\eta)}(j, k, n)=\rho^{(\epsilon, \eta)}_{j,k,n}.\]
\end{enumerate}
\end{proposition}
\begin{proof}
Let us introduce the functions $\widetilde\rho^{ \, (\epsilon, \eta)}_{j,k,n}$: For $j+k+n$ odd, $\epsilon, \eta \in \{0,1\}$ define
\[\widetilde\rho^{ \, (0,0)}_{j,k,n}:=
\begin{cases}
\rho^{(0, 0)}_{j,k,n}& \text{ if } n\geq 0\\
0 & \text{ if } n\leq-1
\end{cases};
\quad\quad\quad\quad
\widetilde\rho^{ \, (1,1)}_{j,k,n}:=
\begin{cases}
\rho^{(1, 1)}_{j,k,n}& \text{ if } n\geq 0\\
0 & \text{ if } n\leq-1
\end{cases};
\]
\[\widetilde\rho^{ \, (0,1)}_{j,k,n}:=
\begin{cases}
\rho^{(0, 1)}_{j,k,n}& \text{ if } n\geq 1\\
0 & \text{ if } n\leq0
\end{cases};
\quad\quad\quad\quad
\widetilde\rho^{ \, (1,0)}_{j,k,n}:=
\begin{cases}
\rho^{(1, 0)}_{j,k,n}& \text{ if } n\geq 1\\
0 & \text{ if } n\leq0
\end{cases}.
\]
Note that functions~$\widetilde\rho^{ \, (\epsilon,\eta)}_{j,k,n}$ satisfy the recurrence relation~\eqref{eq:rec_fund} and the boundary conditions~\eqref{bc}. Consequently,~$\widetilde\rho^{ \, (\epsilon,\eta)}_{j,k,n}$ coincide with the shifted fundamental solution $f_{(\epsilon,\eta)}(j,k,n).$
\end{proof}
Combining the previous proposition with Remark~\ref{rem:density_function} we get the following corollary.
\begin{corollary}\label{cor:f=E}
For $n\geq1$ we have $f_{ \, (\epsilon, \eta)} (j,k,n)=\EE_{j,k,n}\left[1-\mathcal D_{\epsilon, \eta}\right]$.
\end{corollary}

With Proposition~\ref{prop:shift_fund_prob} at hand, we are ready to prove Theorem~\ref{thm}.


\begin{proof}[Proof of Theorem~\ref{thm}]
Combining equations~\eqref{eq:t_n_f_e} and~\eqref{eq:o_n_f_e} with Lemma~\ref{lem:f_E_repr} and Proposition~\ref{prop:shift_fund_prob} yields the result. 
\end{proof}

\bibliographystyle{plain}
\bibliography{bibliotek}

\end{document}